\documentclass{article}

\usepackage{amsmath}
\usepackage{amssymb}
\usepackage{amsthm}
\usepackage{easy-todo}
\usepackage{graphicx}
\usepackage{fullpage}
\usepackage{tikz}
\usetikzlibrary{shapes}
\newtheorem{theorem}{Theorem}
\newtheorem{lemma}{Lemma}
\newtheorem{corollary}{Corollary}

\newtheorem{definition}{Definition}
\newtheorem{reduction}{Reduction}

\newcommand{\Version}{\mathcal{R}}
\newcommand{\Overwrites}{\mathcal{O}}
\newcommand{\Tree}{\mathcal{T}}
\newcommand{\VersionTree}{\mathcal{R}}

\usepackage{blindtext}
\usepackage{graphicx}
\usepackage{subcaption}
\usepackage[utf8]{inputenc}
\usepackage[export]{adjustbox}
\usepackage{wrapfig}
\usepackage{array}
\usepackage{floatrow}
\newcolumntype{s}{>{\columncolor[HTML]{AAACED}} p{3cm}}
\title{Further Unifying the Landscape of Cell Probe Lower Bounds}
\author{Kasper Green Larsen\thanks{Computer Science Department. Aarhus University. \texttt{larsen@cs.au.dk}. Supported by a Villum Young Investigator Grant, a DFF Sapere Aude Research Leader Grant and an AUFF Starting Grant.} \and 
Jonathan Lindegaard Starup\thanks{Computer Science Department. Aarhus University. \texttt{jls@cs.au.dk}} 
\and 
Jesper Steensgaard\thanks{Computer Science Department. Aarhus University. \texttt{steensgaard@cs.au.dk}. }
}

\date{}
\begin{document}
\maketitle

\begin{abstract}
In a landmark paper, P{\v{a}}tra{\c{s}}cu demonstrated how a single lower bound for the static data structure problem of reachability in the butterfly graph, could be used to derive a wealth of new and previous lower bounds via reductions. These lower bounds are tight for numerous static data structure problems. Moreover, he also showed that reachability in the butterfly graph reduces to dynamic marked ancestor, a classic problem used to prove lower bounds for dynamic data structures. Unfortunately, P{\v{a}}tra{\c{s}}cu's reduction to marked ancestor loses a $\lg \lg n$ factor and therefore falls short of fully recovering all the previous dynamic data structure lower bounds that follow from marked ancestor. In this paper, we revisit P{\v{a}}tra{\c{s}}cu's work and give a new lossless reduction to dynamic marked ancestor, thereby establishing reachability in the butterfly graph as a single seed problem from which a range of tight static and dynamic data structure lower bounds follow.
\end{abstract}

\section{Introduction}Proving data structure lower bounds in Yao's cell probe model~\cite{yao:cellprobe} has been an active and important line of research for decades. A data structure in the cell probe model consists of a random access memory, divided into cells of $w$ bits. When answering queries or performing updates, cell probe data structures are only charged for the number of cell accesses (probes) performed. That is, all computation time is for free, and the query and update time is defined solely in terms of the number of probes performed. The cell probe model is very powerful and thus lower bounds for cell probe data structures in particular apply to data structures developed in the standard upper bound model, the word-RAM. Numerous techniques for proving cell probe lower bounds have been developed over the years, ranging from simple reductions from asymmetric communication complexity~\cite{milt:asym}, to tricky round-elimination based techniques~\cite{milt:asym,patrascu06pred,patrascu07randpred,towardsrr}, elegant cell sampling proofs~\cite{pani:metric,larsen:staticloga}, chronograms~\cite{Fredman:chrono}, four-party communication games~\cite{weinstein:fourparty} and combinations of the previous~\cite{larsen:dynamic_count} topped with properties of Chebyshev polynomials and the approximate degree of the AND function~\cite{larsen:crossing}. The range of techniques shows the depth of the field, but may also be intimidating for new researchers that consider entering the field.

\paragraph{Unifying the Field.}
In one of the most beautiful papers on data structure lower bounds~\cite{p2011unifying}, P{\v{a}}tra{\c{s}}cu addressed the aforementioned issue by giving a clean and unified proof of many of the known lower bounds. In his work, P{\v{a}}tra{\c{s}}cu starts at a simple asymmetric communication game termed Lopsided Set-Disjointness (LSD) for which a communication lower bound was already known from earlier work of Miltersen et al.~\cite{milt:asym}. He then proceeds to give a reduction from LSD to reachability data structures in a special graph known as the butterfly graph. In this problem, the input is a directed acyclic graph $G$ with $n$ nodes and $m$ edges. The reduction establishes a lower bound of $t = \Omega(\lg n /\lg w)$ for static data structures that supports reachability queries in $G$ using $m \lg^{O(1)} n$ space using memory cells of $w \geq \lg n$ bits. Here $t$ is the query time and a reachability query is given two nodes $u$ and $v$ in $G$ and must determine whether there is a directed path from $u$ to $v$.

The special structure of the butterfly graph then allowed P{\v{a}}tra{\c{s}}cu to give reductions to numerous classic data structure problems such as 2D range counting, 2D rectangle stabbing and 4D range reporting. These reductions give similar $\lg n/ \lg w$ lower bounds for all these problems while avoiding the heavy machinery often involved in proving lower bounds from scratch. A number of other papers have since then used P{\v{a}}tra{\c{s}}cu's framework and given reductions from either LSD or reachability in butterfly graphs to problems such as 2D skyline counting queries~\cite{brodal:skyline}, approximate distance oracles~\cite{verbin:dist} and range mode queries~\cite{greve:mode}. 

P{\v{a}}tra{\c{s}}cu's work thus reduces all the heavy-lifting involved in proving lower bounds to one initial seed lower bound for LSD and then the rest are reductions, which are familiar to all theoretical computer scientists. Moreover, the initial communication lower bound proof for LSD by Miltersen et al.~\cite{milt:asym} is only a few paragraphs long and easy to grasp if one seeks to understand the whole trail of arguments leading to the lower bounds.

\paragraph{Dynamic Data Structures.}
Common to all of the problems mentioned above is that they are \emph{static} data structure problems. In a static problem, the input is given once and for all and must be preprocessed into a data structure to support queries. In contrast, in a \emph{dynamic} data structure problem, one needs to also support updates to the data. In a classic work, Alstrup et al.~\cite{alstrup1998marked} took an approach similar to P{\v{a}}tra{\c{s}}cu by proving a lower bound for dynamic marked ancestor data structures and then giving reductions to other problems. In the marked ancestor problem, the input is a rooted tree with $n$ nodes. Updates may either mark or unmark the nodes in the tree. A query is specified by a node $u$ in the tree and the goal is to answer whether $u$ has a marked ancestor. Alstrup et al. proved that any data structure for marked ancestor must satisfy $t_q = \Omega(\lg n/ \lg(t_u w \lg n))$, where $t_q$ is the query time and $t_u$ the update time. While the marked ancestor problem may seem abstract, Alstrup et al.~\cite{alstrup1998marked} gave reductions from marked ancestor to numerous dynamic data structure problems such as dynamic 2D range emptiness, union-find and dynamic connectivity, thereby establishing similar lower bounds for all these problems. See Alstrup et al.~\cite{alstrup1998marked} for further details.

In P{\v{a}}tra{\c{s}}cu's unifying work~\cite{p2011unifying}, he also demonstrated that his static lower bound for reachability in the butterfly graph \emph{almost} implies the marked ancestor lower bound by Alstrup et al.~\cite{alstrup1998marked}. That is, he gave a reduction from reachability in the butterfly graph to dynamic marked ancestor. His reduction establishes a $t_q = \Omega(\lg n/(\lg(t_u w \lg n) \lg \lg n))$ lower bound for marked ancestor, and thus also for the whole range of dynamic problems where we already had reductions from marked ancestor. Thus, except for a $\lg \lg n$ factor, the hardness of reachability in the butterfly graph single-handily explains the hardness of a wealth of static and dynamic data structure problems. If only the reduction did not lose a $\lg \lg n$ factor!

\paragraph{Our Contribution.}
In this work, we complete P{\v{a}}tra{\c{s}}cu's unifying work by demonstrating a lossless reduction from reachability in the butterfly graph to dynamic marked ancestor, thereby establishing the tight $t_q = \Omega(\lg n/\lg(t_u w \lg n))$ lower bound for dynamic marked ancestor. Thus one single lower bound suffices as a seed for the whole range of reductions among static and dynamic data structure problems. See Figure~\ref{fig:reduc_overview} for an overview of some of the lower bounds that follow from this reduction to dynamic marked ancestor.
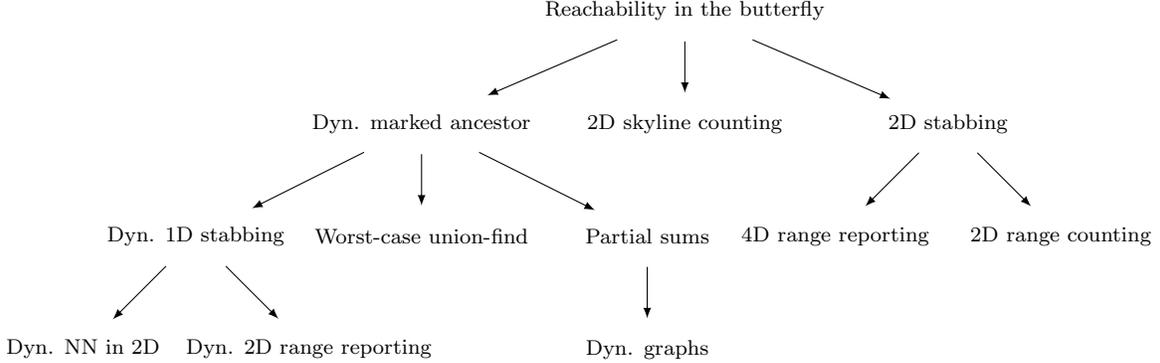
\begin{figure*}
\begin{tikzpicture}[
level distance=1.5cm,
level 1/.style={sibling distance=3.5cm},
level 2/.style={sibling distance=3cm},
level 3/.style={sibling distance=3cm},
edge from parent/.style={draw,-latex}]
\tikzstyle{every node}=[ellipse,draw,
minimum height=0.8cm,minimum width=1.5cm]
\node(Root)[draw=none]{\footnotesize Reachability in the butterfly}
child{node[draw=none]{\footnotesize Dyn. marked ancestor} 
child{node[draw=none]{\footnotesize Dyn. 1D stabbing}
  child{node[draw=none]{\footnotesize Dyn. NN in 2D}}
  child{node[draw=none]{\footnotesize Dyn. 2D range reporting}}
}
child{node[draw=none]{\footnotesize Worst-case union-find}}
child{node[draw=none]{\footnotesize Partial sums}
child{node[draw=none]{\footnotesize Dyn. graphs}}
}
}
child{node[draw=none]{\footnotesize 2D skyline counting}}
child{node[draw=none]{\footnotesize 2D stabbing}
child{node[draw=none]{\footnotesize 4D range reporting}}
child{node[draw=none]{\footnotesize 2D range counting}}};
\end{tikzpicture}
\caption{Some of the lower bounds obtained via reductions from reachability in the butterfly graph. The lower bound for dynamic marked ancestor, and all problems in its subtree, become $t_q = \Omega(\lg n/\lg(t_u w \lg n))$ where $t_q$ is the worst case query time and $t_u$ the worst case update time. For references to original papers containing these reductions, we refer the reader to P{\v{a}}tra{\c{s}}cu~\cite{p2011unifying}.}
 \label{fig:reduc_overview}
\end{figure*}

\subsection{Proof Overview}
In the following, we sketch the key idea underlying our reduction from reachability in the butterfly graph to dynamic marked ancestor. To explain it, we start by sketching P{\v{a}}tra{\c{s}}cu's original reduction. The basic idea in P{\v{a}}tra{\c{s}}cu's work, which loses a $\lg \lg n$ factor, is to take a data structure for dynamic marked ancestor and apply the classic technique of \emph{full persistence} to it. A fully persistent data structure, is a dynamic data structure, along with a rooted version tree. Each node of the version tree contains a sequence of updates. A query to the data structure is also given a node $v$ of the version tree, and the data structure must answer the query as if precisely the updates on the path from the root to $v$ had been performed. P{\v{a}}tra{\c{s}}cu argued that \emph{any} dynamic data structure can be made fully persistent at the cost of a $\lg \lg n$ factor in query time, while using space $O(n t_u)$ when the version tree contains a total of $n$ nodes. Secondly, P{\v{a}}tra{\c{s}}cu showed that a fully persistent data structure for dynamic marked ancestor may be used to solve reachability in the butterfly graph. It is thus the application of full persistence that causes the $\lg \lg n$ factor loss in the lower bound.

\paragraph{Non-Determinism.}
Our key idea is to use a slightly stronger lower bound for reachability in the butterfly graph as a starting point. Interesting work by Yin~\cite{yin2010cell} and by Wang and Yin~\cite{wang2014certificates} studied the notion of \emph{non-deterministic} data structures. Non-deterministic data structures may essentially \emph{guess} which memory cells to read when answering a query. More formally, one can think of a prover that specifies a set of memory cells to a verifier. The verifier receives that set of memory cells and must either output the correct answer to the query, or reject the set of cells. Similarly to the complexity class NP, we require that there must be a set of cells resulting in the verifier answering the query, i.e. there must exist a \emph{certificate}. Wang and Yin~\cite{wang2014certificates} showed that P{\v{a}}tra{\c{s}}cu's lower bound for reachability in butterfly graphs also holds for non-deterministic data structures. 

While lower bounds for non-deterministic data structures may appear rather abstract, we elegantly exploit precisely this non-determinism to avoid the $\lg \lg n$ loss in P{\v{a}}tra{\c{s}}cu's reduction to marked ancestor. More concretely, we show that any dynamic data structure can be made fully persistent at an $O(1)$ factor increase in query time if we at the same time make it non-deterministic. In light of Wang and Yin's lower bound for non-deterministic reachability in the butterfly graph, this is perfectly fine and we obtain the tight lower bound for dynamic marked ancestor.

\section{Certificates for Fully Persistent Data Structures}In this section, we formally define non-deterministic data structures, or equivalently, certificates in data structures. We then demonstrate how to make any dynamic data structure fully persistent using non-determinism.

Let $\mathcal{D}$ be a set of databases, $\mathcal{Q}$ a set of queries, and $\mathcal{Z}$ a set of results. A function $f : \mathcal{Q} \times \mathcal{D} \to \mathcal{Z}$ specifies a \emph{data structure problem}. The result of a query $q\in \mathcal{Q}$ on database $d \in \mathcal{D}$ is $f(q, d)$. The problem $f$ has $(s, w, t)$-certificates if there is a code $T : \mathcal{D} \to (\{0,1\}^w)^s$ such that any query on $d \in \mathcal{D}$ can be answered from some $t$ elements of the tuple $T(d)$. We think of $T(d)$ as a \emph{table} of $s$ \emph{cells} each of which has $w$ bits. We denote by $T_d(i)$ the $i$'th cell of $T(d)$ and for $P = \{ i_1, ..., i_k \}$ we denote by $T_d(P)$ the sequence $(i_1, T_d(i_1)), ..., (i_k, T_d(i_k))$. More formally, we have:

\begin{definition}
\label{def:certificate}
A data structure problem $f : \mathcal{Q} \times \mathcal{D} \to \mathcal{Z}$ has $(s, w, t)$-certificates if for some code $T: \mathcal{D} \to (\{0,1\}^w)^s$ there exists a verifier $V$ such that for all $q \in \mathcal{Q}$, $d \in \mathcal{D}$, and $P \subseteq \{ 1, ..., s\}$ with $|P|=t$ either $V(q, T_d(P)) = f(q,d)$ or $V(q, T_d(P)) = \bot$, and for at least one such $P$ we have $V(q, T_d(P)) = f(q, d)$. The symbol $\bot \notin \mathcal{Z}$ indicates verification failure.
\end{definition}

A data structure problem having $(s, w, t)$-certificates corresponds to the existence of a non-deterministic data structure using space $s$ cells of $w$ bits, that can answer any query by guessing $t$ cells to look at. We can use this fact to show the existence of certificates by describing a non-deterministic data structure. We remark that lower bounds for $(s, w, t)$-certificates also hold for deterministic data structures as one can obtain a verifier from a deterministic data structure by simply running the query algorithm of the data structure, and if it ever requests a cell not in $P$, we return $\bot$.

To illustrate the power of non-determinism and certificates, and the key to our improved reduction, let us define the rank problem. The \textit{\textbf{rank problem}} for some universe $[U]$, is the data structure problem defined by the function $f : [U] \times 2^{[U]} \to \mathbb{N}$ where $f(x, S) = |\{s \in S : s \le x \}|$ . Given a set of integers $S$ with $|S|=n$ (the database) and an integer $x$ (the query), we are interested in the \emph{rank} of $x$ in $S$, i.e., the number of elements in $S$ that are not greater than $x$. The rank problem is at least as hard as \emph{predecessor search}, and thus has a lower bound of $t = \Omega(\min\{\lg \lg U, \lg_w n\})$ due to the work by P{\v{a}}tra{\c{s}}cu and Thorup~\cite{patrascu06pred,patrascu07randpred}. Moreover, in P{\v{a}}tra{\c{s}}cu's reduction from reachability in the butterfly graph to dynamic marked ancestor, his application of full persistence needs to solve predecessor search on a universe of size $U = n$ and thus costs $\Omega(\lg \lg n)$. We avoid this by demonstrating that the rank problem can be solved much more efficiently if we allow non-determinism:

\begin{lemma}
\label{prop:rank_certificate}
The rank problem has $(n, w, 2)$-certificates for any $w > \lg U$.
\end{lemma}
\begin{proof}
Given $S \subseteq [U]$ we store a sorted list of the elements in $S$. More precisely, we store the table $T_S$ whose $i$'th entry is the $i$'th smallest element of $S$. This table makes it possible to verify the rank of a query $x$ by looking at the element less than and greater than $x$: We define a verifier $V$, such that for any $x \in [U]$ and $P\subseteq \{ 1, ..., s\}$ where $|P|=2$ and $T_S(P) = (i, \ell), (j, u)$, we let $V(x, T_S(P)) = i$ if $i+1 = j$ and $\ell \le x < u$. If $i=1$ and $x < \ell$ let $V(x, T_S(P)) = 0$. If $j=n$ and $x \ge u$ let $V(x, T_S(P))=n$. Otherwise, $V(x, T_S(P)) = \bot$. The table $T_S$ has $n$ entries, each with $w$ bits and the verifier needs $t=2$ cells to answer the query, i.e. the rank problem has $(n,w,2)$-certificates.
\end{proof}

As mentioned earlier, we use efficient certificates for the rank problem to make any dynamic data structure fully persistent. We first formally define the (static) fully persistent version of a data structure problem.
 
\begin{definition}
Given some dynamic problem, its static, fully persistent version is a data structure problem on a rooted tree. Every node has a sequence of update operations. Queries are pairs $(q, u)$ and the answer is the result of executing $q$ after the sequence of updates found on the path from the root to $u$.
\end{definition}

We call this tree the \emph{version tree}. A solution to a dynamic problem implies certificates for its static, fully persistent version. Specifically,

\begin{theorem}
\label{thm:fully_persistent_certificate}
If the dynamic version of a problem has a solution with update time $t_u$ and query time $t_q$, then the static, fully persistent version with $m$ updates has $(s, w, t)$-certificates for $s = O(m \cdot t_u)$ and $t = O(t_q)$ provided that $w = \Omega(\lg m)$.
\end{theorem}
\begin{proof}
Let $\Version$ be the version tree corresponding to the static, fully persistent version of the dynamic problem. Consider performing a depth-first traversal of $\Version$. When discovering a node $u$ during the traversal, run the sequence of updates in $u$ on a dynamic data structure with update time $t_u$ and query time $t_q$. For each cell that has its contents overwritten, record the changes. When finishing a node $u$ during the traversal, revert all changes made to cell contents. To answer a query $(q,u)$ to the static fully persistent problem, we simulate the query algorithm of the dynamic data structure on the query $q$. Each time the query algorithm requests a memory cell $c$, we need to retrieve its contents as it was precisely after discovering $u$ during the depth-first traversal. If we can do so for all cells $c$, we will answer the query correctly. We call the contents of $c$ at the discovery time of $u$ the \emph{contents of $c$ at $u$}.

A simple, but inefficient, solution would be to store a table for each memory cell $c$, having one entry per discovery and finishing time of a node $u \in \Version$. In each entry, we store the contents of $c$ at that time during the depth-first traversal. Furthermore, we could store an auxiliary table indexed by the nodes of $\Version$. The entry corresponding to a node $u$ stores the discovery time of $u$ (the auxiliary table could be a hash table if nodes $u \in \Version$ are not specified by consecutive integers). Then, given a cell $c$, we could retrieve its contents at $u$ simply by looking up its discovery time in the auxiliary table and then looking up its contents in the table for $c$. The problem with this solution is that the space usage can be as large as $\Omega(m |\Version| t_u)$ if there are a total of $m$ updates in nodes of $\Version$ (the memory may have $m t_u$ cells, and for each, we store a table with $\Omega(|\Version|)$ entries).

Our goal is to reduce the space usage of the simple solution above by storing a much smaller table for each cell $c$. Concretely, consider a cell $c$ and let $\Overwrites_c \subseteq \Version$ denote the subset of nodes whose updates change the contents of $c$. During the depth-first traversal, the contents of $c$ only change at discovery and finishing times of nodes in $\Overwrites_c$. Let $S_c$ denote this set of discovery and finishing times and store a table with entries $T(c,1),\dots,T(c,2|S_c|)$ having one entry per event (discovery or finishing) in $S_c$. The $i$'th entry $T(c,i)$ stores the contents of $c$ immediately after the $i$'th event in $S_c$. Now, given a node $u$ with discovery time $d_u$, let $e \in S_c$ denote the largest time in $S_c$ that is less than or equal to $d_u$. That is, $e$ is the \emph{predecessor} of $d_u$ in $S_c$. The contents of $c$ at $u$ equals the contents of $c$ at time $e$ during the depth-first traversal. See Figure~\ref{fig:dfs} for an illustration.
\begin{figure}[h]
\centering
\begin{tikzpicture}[
level distance=1cm,
level 1/.style={sibling distance=1.75cm},
level 2/.style={sibling distance=1cm}]
\tikzstyle{every node}=[circle,draw]
\node[fill,label=above:{$1,8$}](Root){}
child[missing]{}
child{node[label=above:{$2,7$}]{}
child{node[fill,label=above:{$3,4$}]{}}
child{node[draw, double, label=above:{$5,6$}]{}}};
\end{tikzpicture}
\caption{Assume the two black nodes writes to a cell $c$. The node with discovery time $1$ writes $x$ to $c$ and the node with discovery time $3$ writes $y$. Initially, assume $c$ has contents $z$. During a depth-first traversal, the contents of $c$ change as follows: Initially it is $z$, at time $1$ it becomes $x$, at time $3$ it becomes $y$, at time $4$ it becomes $x$ and at time $8$ it becomes $z$. Assume we are interested in the contents of $c$ at the discovery of the double-circled node (at time $5$). We have $S_c = \{ 1, 3, 4, 8 \}$. The predecessor of $5$ in $S_c$ is $4$ and the contents of $c$ at time $5$ equals the contents at time $4$, namely $x$.}
\label{fig:dfs}
\end{figure}
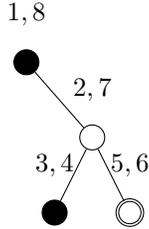
Let $r_e$ denote the \emph{rank} of $e$ in $S_c$, then $T(c,r_e)$ stores those contents. Since $r_e$ equals the rank of $d_u$ in $S_c$, we thus need to solve a rank query on $S_c$. Fortunately, we have already seen in Lemma~\ref{prop:rank_certificate} that such rank queries have $(|S_c|,w,2)$-certificates for $w > \lg U$ where $U$ denotes the largest event time in $S_c$. That is, we can non-deterministically retrieve the contents of $c$ at $u$ in $O(1)$ time using $O(|S_c|)$ space for memory cell $c$. Summing over all memory cells, the total space usage is $O(\sum_c |S_c|) = O(m t_u)$ (there are $m$ updates and each may change up to $t_u$ cells) and the total query time needed to simulate a query is $O(t_q)$ ($O(1)$ time per simulated cell access). Since the largest event time is bounded by $2|\Version| \leq 2m$, we require $w = \Omega(\lg m)$. We also need the auxiliary table mapping nodes to discovery times. Using a hash table, this can be done in worst case $O(1)$ time and $O(|\Version|) = O(m)$ space (using e.g. Cuckoo Hashing~\cite{cuckoo}). In summary, we have given $(O(mt_u), w, O(t_q))$-certificates for $w = \Omega(\lg m)$.
\end{proof}

As already discussed in the introduction, this result can be used to prove lower bounds on dynamic problems. Concretely, a lower bound on the certificates of the static, fully persistent version of a problem, gives a lower bounds for the dynamic version as well. The next section illustrates this for the dynamic marked ancestor problem.

\section{Lower Bound for Dynamic Marked Ancestor}
\label{sec:dyn_marked_anc}
The reduction we give from reachability in the butterfly graph to dynamic marked ancestor is due to P{\v{a}}tra{\c{s}}cu, only that we have improved one step of the reduction. More concretely, P{\v{a}}tra{\c{s}}cu reduces reachability in the butterfly graph to static fully persistent marked ancestor and then from there to dynamic marked ancestor. The last step of his reduction losses a $\lg \lg n$, whereas we have seen in Theorem~\ref{thm:fully_persistent_certificate} that this step can be performed at an $O(1)$ factor change in query time if we allow non-determinism. For completeness, we have chosen to include P{\v{a}}tra{\c{s}}cu's reduction from reachability in the butterfly graph to static fully persistent marked ancestor.

Recall that in the marked ancestor problem, the input is a rooted tree $\Tree$ where updates either mark or unmark nodes. A query asks whether a given node has a marked ancestor. In the static fully persistent version of the problem, we are given as input a version tree $\VersionTree$ with such updates. Queries are pairs $(u,v)$ with $u \in \Tree$ and $v \in \VersionTree$. Such a query must output whether $u$ has a marked ancestor if performing precisely the updates in the nodes on the path from the root of $\VersionTree$ to $v$.

To present P{\v{a}}tra{\c{s}}cu's reduction, we start by describing the butterfly graph. A butterfly is a graph specified by a degree $b$ and depth $d$. The graph has $d+1$ layers, each of which has $b^d$ vertices. Viewing the vertices of layer $i$ as vectors in $[b]^d$ (or numbers in base $b$), there is an edge from a vertex in layer $i$ to a vertex in layer $i+1$, precisely if their vectors are equal in all coordinates, except possibly coordinate $i$. The vertices in layer $0$ are called \emph{sources} and the vertices in layer $d$ are called \emph{sinks}. In the problem of reachability in the butterfly graph, one is given as input a subset of the edges in a butterfly graph. A query is specified by a source-sink pair, and the goal is to return whether the source can reach the sink. Notice that in the full butterfly graph, there is precisely one path from each source to each sink, namely the path that starts in the source, and in layer $i$, it takes the edge leading to the neighbouring vertex in layer $i+1$ whose $i$'th coordinate equals the $i$'th coordinate of the sink. The path thus "morphs" the coordinates of the source into those of the sink, one coordinate at a time. For a reachability query on a subgraph, we thus have to determine whether at least one edge on this unique path is missing or not.

An example of a subgraph of a butterfly with depth 2 and degree 2 is shown in Figure \ref{fig:reduction}(a).

P{\v{a}}tra{\c{s}}cu gave the following reduction \cite{p2011unifying}:

\begin{figure*}
\begin{subfigure}[t]{0.45\textwidth}
\centering
\begin{tikzpicture}
\tikzset{dot/.style={fill=black,circle,scale=0.6}}
\node[dot] at (0,2){};
\node[dot] at (1,2){};
\node[dot] at (2,2){};
\node[dot] at (3,2){};
\node[dot,label=above:$s_1$] at (0,4){};
\node[dot,label=above:$s_2$] at (1,4){};
\node[dot,label=above:$s_3$] at (2,4){};
\node[dot,label=above:$s_4$] at (3,4){};
\node[dot,label=below:$t_1$] at (0,0){};
\node[dot,label=below:$t_2$] at (1,0){};
\node[dot,label=below:$t_3$] at (2,0){};
\node[dot,label=below:$t_4$] at (3,0){};
\draw(0,2)--(1,4)--(1,2)--(3,0)--(3,4)--(2,2)--(2,0)--(0,2)--(0,4);
\draw(2,4)--(3,2);
\draw(2,2)--(0,0);
\draw[loosely dashed](0,4)--node[left]{$e_1$}(1,2)
--node[left]{$e_4$}(1,0)--node[left]{$e_5$}(3,2);
\draw[loosely dashed](2,2)--node[left]{$e_2$}(2,4);
\draw[loosely dashed](0,0)--node[left]{$e_3$}(0,2);
\end{tikzpicture}
\caption{A subgraph $G'$ of a butterfly with degree 2 and depth 2. Dashed lines indicate missing edges which are named $e_1, ..., e_5$}
\vspace{0.5cm}
\end{subfigure}
\hfill
\begin{subfigure}[t]{0.45\textwidth}
\centering
\begin{tikzpicture}
\tikzset{dot/.style={fill=black,circle,scale=0.6}}
\node[dot] at (0,2){};
\node[dot] at (1,2){};
\node[] at (2,2){};
\node[] at (3,2){};
\node[dot,label=above:$s_1$] at (0,4){};
\node[] at (1,4){};
\node[] at (2,4){};
\node[] at (3,4){};
\node[dot,label=below:$t_1$] at (0,0){};
\node[dot,label=below:$t_2$] at (1,0){};
\node[dot,label=below:$t_3$] at (2,0){};
\node[dot,label=below:$t_4$] at (3,0){};
\draw(1,2)--(3,0);
\draw(0,2)--(0,4);
\draw(0,2)--(2,0);
\draw[loosely dashed](0,4)--node[left]{$e_1$}(1,2)
--node[left]{$e_4$}(1,0);
\draw[loosely dashed](0,0)--node[left]{$e_3$}(0,2);
\end{tikzpicture}
\caption{The tree $S_1'$ rooted at source $s_1$.}
\vspace{0.5cm}
\end{subfigure}
\begin{subfigure}[t]{0.5\textwidth}
\centering
\begin{tikzpicture}[
level distance=1.5cm,
level 1/.style={sibling distance=4cm},
level 2/.style={sibling distance=2cm}]
\tikzstyle{every node}=[ellipse,draw,
minimum height=0.8cm,minimum width=1.5cm]
\node(Root){}
child{node {$e_3, e_4$} 
child{node[label=below:$s_1$]{$e_1$}}
child{node[label=below:$s_2$]{}}}
child{node {$e_5$}
child{node[label=below:$s_3$]{$e_2$}}
child{node[label=below:$s_4$]{}}};
\end{tikzpicture}
\caption{The corresponding version tree. Consider the left child of the root. The sources of the butterfly that are descendants of this node are $s_1$ and $s_2$. There are four edges that can be reached by precisely $s_1$ and $s_2$ in the butterfly. Of these $e_3$ and $e_4$ are missing.}
\end{subfigure}
\hfill
\begin{subfigure}[t]{0.45\textwidth}
\centering
\begin{tikzpicture}[
level distance=1cm,
level 1/.style={sibling distance=1.75cm},
level 2/.style={sibling distance=1cm}]
\tikzstyle{every node}=[circle,draw]
\node(Root){}
child{node{} 
child[dashed]{node[fill,label=below:$t_1$]{}}
child{node[label=below:$t_3$]{}}}
child[dashed]{node[fill]{}
child[dashed]{node[fill,label=below:$t_2$]{}}
child[solid]{node[label=below:$t_4$]{}}};
\end{tikzpicture}
\caption{The marked ancestor tree after the updates specified by $s_1$ of the version tree. The dashed edges correspond to the missing edges $e_1$, $e_3$ and $e_4$. The corresponding lower endpoints have been marked.}
\end{subfigure}
\caption{Illustration of reduction. As an example, consider the edge $e_1$ in (a). It starts at the node with vector $v_\ell = (0,0)$ at layer $i=0$ and goes to the node with vector $v_u=(1,0)$ at layer $i+1=1$, see (a). It is represented by an update in the version tree node with index $\sum_{k=0}^{d-i-1} b^k v_{\ell}[i+k]=\sum_{k=0}^{2-0-1} b^k v_{\ell}[0+k] = 0 \cdot b^0 + 0 \cdot b^1=0$ in layer $d-i=2-0=2$, see (c). It marks the node at index $\sum_{k=0}^{i} b^{i-k} v_u[k]=\sum_{k=0}^{0} b^{0-k} v_u[k] = 1\cdot b^0=1$ in layer $i+1 = 1$ of $\Tree$, see (d).}
\label{fig:reduction}
\end{figure*}
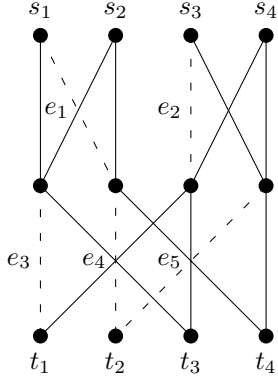
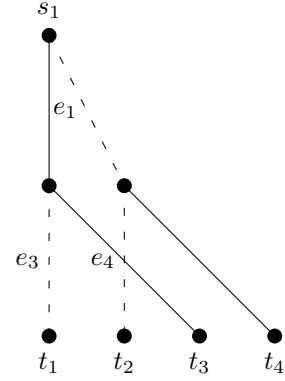
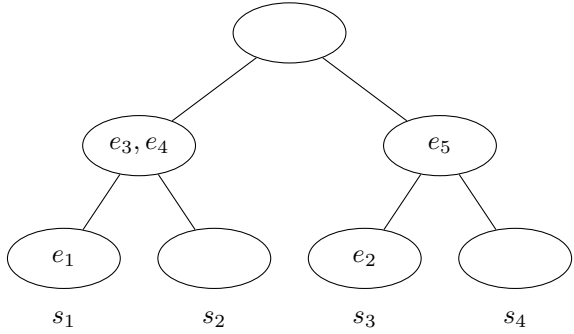
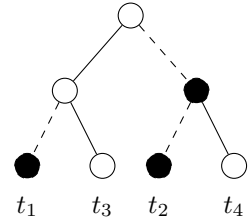

\begin{reduction}
\label{red:reachability_to_marked_ancestor}
Let $G$ be a butterfly with $m$ edges.
The reachability problem on subgraphs of $G$ reduces to the static fully persistent version of the marked ancestor problem with $O(m)$ updates.
\end{reduction}
\begin{proof}
We are given as input a subgraph $G'$ of the butterfly $G$ with degree $b$ and depth $d$. Define for each source $s_i$ the tree $S_i$ consisting of everything reachable from $s_i$ in $G$. Similarly, define $S_i'$ as the subgraph of $S_i$, where only edges from $G'$ remain. The tree $S_i$ has degree $b$ and depth $d$ and the leaves correspond to the sinks in $G$. A source $s_i$ can reach a sink $t_j$ in $G'$ precisely if there is a path from the root of $S_i'$ to the leaf corresponding to $t_j$. See Figure~\ref{fig:reduction}(b). The first idea in the reduction, is that we will use markings on a tree $\Tree$ of degree $b$ and depth $d$ to denote missing edges in $S_i'$. Concretely, if some edge from $S_i$ is missing in $S_i'$, we will mark the lower endpoint of that edge in $\Tree$. Then there is a path from $s_i$ to $t_j$ in $G'$ if and only if there are no marked vertices on the path from the root of $\Tree$ to the leaf corresponding to $t_j$.

The second idea is to use full persistence to represent all the different trees $S_i'$ as different versions of the same marked ancestor tree $\Tree$. Thus we have the marked ancestor tree $\Tree$ and a version tree $\VersionTree$, each of which is a complete tree with degree $b$ and depth $d$. The leaves of $\VersionTree$ correspond to the sources of $G$ and the leaves of $\Tree$ correspond to the sinks. We will assign updates to the version tree such that, if one performs the updates on the path from the root of $\VersionTree$ to a leaf corresponding to some source $s_i$, then $\Tree$ looks exactly like $S_i'$ (markings at lower endpoints of edges missing in $S_i'$). See Figure~\ref{fig:reduction}(d) for an illustration. 

We need to represent missing edges in $G'$ by mark operation in nodes of $\VersionTree$. Consider a missing edge $e=(\ell,u)$ from layer $i$ to $i+1$ in $G$. Let $v_a \in [b]^d$ denote the vector representing the node $\ell$ (its index into layer $i$, written in base $b$, with $v_a[0]$ being the least significant digit). By definition of the butterfly, precisely sources $s_i$ with vectors of the form $** \cdots **v_{\ell}[i] v_\ell[i+1] \cdots v_\ell[d-1]$ can reach $\ell$, where $*$ can be any value in $[b]$. These sources are precisely the set of all leaves in the subtree rooted at the node of index $\sum_{k=0}^{d-i-1} b^k v_\ell[i+k]$ into layer $d-i$ of $\VersionTree$. We will thus place the mark operation in that node, ensuring that the operation will be performed precisely when we query with a source that can reach $\ell$.

We now need to determine which edge of $\Tree$ to mark, or technically, which lower endpoint of an edge to mark. For this, observe that if $s_j$ is any source that can reach $\ell$, then the path from $s_j$ to $\ell$ has the same form regardless of $s_j$: At layer $i$, take the step to the $v_\ell[i]$'th child/neighbour at layer $i+1$. That is, for any $s_j$ that can reach $\ell$, we have that $e$ is exactly the edge from the node with index $\sum_{k=0}^{i-1} b^{i-1-k} v_\ell[k]$ in layer $i$ of $S_j'$ to the node with index $\sum_{k=0}^{i} b^{i-k} v_u[k]$ in layer $i+1$ of $S_j'$. To summarize, we represent the missing edge $e=(\ell,u)$ by marking the node of index $\sum_{k=0}^{i} b^{i-k} v_u[k]$ in layer $i+1$ of $\Tree$ using an update operation in the node of index $\sum_{k=0}^{d-i-1} b^k v_\ell[i+k]$ into layer $d-i$ of $\VersionTree$. See Figure~\ref{fig:reduction} for an illustration.

As already described above, we can now answer a reachability query from a source $s_i$ to a sink $t_j$ in $G'$ by asking the marked ancestor query using the leaf corresponding to $t_j$ in $\Tree$ on the version node in $\VersionTree$ corresponding to $s_i$.
\end{proof}

Now that we have established the reduction from reachability in the butterfly to static fully persistent marked ancestor, we can use the following lower bound from \cite{wang2014certificates}:

\begin{theorem}
\label{thm:reachability_oracle_lower_bound}
If reachability in subgraphs of a butterfly with $n$ edges has $(s, w, t)$-certificates for $s = \Omega(n)$, then $t = \Omega(\lg n / \lg \frac{sw}{n})$.
\end{theorem}

Let us combine it all to derive the lower bound for dynamic marked ancestor. Consider a solution to dynamic marked ancestor with query time $t_q$ and update time $t_u$. Given a subgraph $G'$ of a butterfly $G$ with $n$ edges, we can use Reduction \ref{red:reachability_to_marked_ancestor} to solve reachability in $G'$ with $O(n)$ updates to the static fully persistent marked ancestor problem. Using Theorem \ref{thm:fully_persistent_certificate}, the dynamic marked ancestor solution gives us $(O(n t_u), w, O(t_q))$-certificates for this whenever $w = \Omega(\lg n)$. Theorem~\ref{thm:reachability_oracle_lower_bound} finally gives us that $t = \Omega(\lg n/\lg(t_u w))$ for any $w = \Omega(\lg n)$ and we conclude for any cell size $w$:
\begin{corollary}
\label{cor:marked_ancestor_lower_bound}
The dynamic marked ancestor problem requires query time $t_q = \Omega(\lg n /\lg (t_u w \lg n))$ on $n$-vertex trees.
\end{corollary}
This lower bound matches the original bound from \cite{alstrup1998marked}.
\section{Acknowledgements}This paper was written as part of the Talent Track Program\footnote{http://studerende.au.dk/studier/fagportaler/datalogi/ undervisning/talentforloeb/} for Bachelor students at the Department of Computer Science at Aarhus Univerity, Denmark.

\bibliography{bib}{}
\bibliographystyle{plain}

\end{document}